\documentclass[12pt]{amsart}
\usepackage{amsmath, tikz-cd}
\usepackage{amsfonts}
\usepackage{amsthm}
\usepackage{bbm}
\usepackage{mathtools}
\usepackage{xcolor} 
\allowdisplaybreaks 
\usepackage{amscd,amssymb,latexsym}
\usepackage[mathscr]{euscript}
\usepackage{epsfig}
\usepackage{color}

\allowdisplaybreaks 

\newtheorem{Thm}{Theorem}[section]
\newtheorem{Cor}[Thm]{Corollary}
\newtheorem{Lem}[Thm]{Lemma}

\newtheorem{Rem}[Thm]{Remark}
\newtheorem{Ex}[Thm]{Example}

\newtheorem{Not}[Thm]{Notation}

\newcommand{\N}{\mathbb{N}}

\newcommand{\R}{\mathbb{R}}
\newcommand{\E}{\mathbb{E}}
\newcommand{\C}{\mathbb{C}}

\newcommand{\pr}{\text{Pr}\,}

\DeclareMathOperator{\argmin}{argmin}

\newcommand{\1}{\mathbbm 1}

\newcommand{\ds}{\displaystyle}

\newcommand{\ra}{\rightarrow}

\def\tr{\operatorname{tr}\, }

\def\span{\operatorname{span}}
\newcommand{\ten}{\otimes}

\newcommand{\mc}{\mathcal}

\newcommand{\bma}{\boldmath}

\newcommand{\X}{\mathbf{X}}
\newcommand{\Y}{\mathbf{Y}}
\newcommand{\opt}{\text{opt}}
\newcommand{\lmax}{\ell_{\text{max}}}

\newcommand{\final}{\text{final}\,}
\newcommand{\pruned}{\text{pruned}\,}

\begin{document}

\title[Quantum Data Compression using Dynamical Entropy]{Optimality in Quantum Data Compression using Dynamical Entropy}

\author{George Androulakis}
\email{giorgis@math.sc.edu}	

\author{Duncan Wright}
\thanks{This work is part of the second author's PhD thesis.}
\thanks{The second author was partially supported by a SPARC Graduate Research Grant from the Office of the Vice President for Research at the University of South Carolina.}
\email{dwright2@wpi.edu}

\date{\today}

\begin{abstract}
In this article we study lossless compression of strings of pure quantum states of indeterminate-length quantum codes which were introduced by Schumacher and Westmoreland. 
Past work has assumed that the strings of quantum data are prepared to be encoded in an independent and identically distributed way. 
We introduce the notion of quantum stochastic ensembles, allowing us to consider strings of quantum states prepared in a general way. 
For any quantum stochastic ensemble we define an associated quantum dynamical system and prove that the optimal average codeword length via lossless coding is equal to the quantum dynamical entropy of the  associated quantum dynamical system. 
\end{abstract}

\maketitle

\section{Introduction}

In the theory of data compression of classical information theory one wishes to encode a symbol set, $S$, with a code, $C$, which is a mapping from the symbol set $S$ to the set  $A^+$ of all finite strings (or sequences) of elements from the alphabet $A$,  where $A$ is usually taken to be the binary alphabet $\{ 0, 1\}$. 
The set $A^+$ is frequently referred to as the codebook and its elements are called codewords. 
Since we compress long strings (sequences) of messages, concatenation is used to extend the code $C$ to the set $S^+$ containing all finite strings from the symbol set $S$. 
This extension of $C$ is denoted by $C^+$ and it is called the extended code. 
A code $C$ is said to be uniquely decodable if its extended code is an injective function. 
In that case, the decoding function is the inverse of $C^+$. 
If each symbol $x$ of the symbol set $S$ that we wish to encode is always prepared with the same probability $p(x)$, independent of the string of symbols that have appeared earlier, then the sequence 
$(X_n)_{n \in \mathbb{N}}$ of random variables which gives us the string of symbols to be encoded 
is independent and identically distributed (i.i.d.) with values in the symbol set $S$ with probability mass function equal to 
$(p(x))_{x \in S}$. 
If $X$ denotes any member of this sequence of random variables then its Shannon entropy $H(X)$ is defined as 
$$
H(X)= - \sum_{x \in S} p(x) \log_2 p(x) .
$$

In Shannon's original works on the subject (\cite{Shannon48a, Shannon48b}), the Noiseless Coding Theorem was proved which states that, for any $\delta>0$, $(H(X)+\delta)$-many binary bits per symbol are sufficient in order to encode strings of symbols if each entry of the sequence is prepared in a i.i.d.\
 way, with probability of error tending to zero as the length of the strings tend to infinity. 
Moreover, Shannon showed that for any $R<H(X)$, if at most $R$ bits are used per symbol, then the probability of error tends to 1 as the length of the strings tend to infinity. 
Thus the Shannon entropy $H(X)$ can be interpreted as the 
minimum expected number of binary bits per symbol that are necessary in order to encode strings of symbols with arbitrarily small error (i.e. \ asymptotically lossless coding)
given that the elements of the string of symbols are encoded in an i.i.d.\ way.

The setting of quantum data compression for indeterminate-length quantum codes is similar to the setting of classical data compression. In this case, the symbol set $\mathcal{S}$ contains the symbol states which are normalized vectors spanning a Hilbert space $H_\mathcal{S}$. 
Here we only consider the compression of pure quantum states, therefore we restrict our attention to 
normalized vectors or pure states. 
The classical binary alphabet $A=\{0,1\}$ is replaced by the set of qubits $\mathcal{A}= \{ |0 \rangle , | 1\rangle \}$ which is the standard orthonormal basis of the Hilbert space $\mathcal{H}_\mathcal{A}=\mathbb{C}^2$. 
The classical codebook $A^+$ is replaced by the free Fock space $H_\mathcal{A}^\oplus = \oplus_{\ell=0}^\infty H_\mathcal{A}^{\otimes \ell}$. 
A quantum code is a linear isometry $U: H_\mathcal{S} \to H_\mathcal{A}^\oplus$, and the corresponding extended code is a map $U^+$ which is defined on the free Fock space 
$H_\mathcal{S}^\oplus=\oplus_{\ell=0}^\infty H_\mathcal{S}^{\otimes \ell}$ by ``concatenation;" i.e.\ tensor 
products of the values of $U$ in the free Fock space 
$H_\mathcal{A}^\oplus$. 
The quantum code $U$ is called uniquely decodable if $U^+$ is also an isometry.

The Noiseless Coding Theorem was extended to indeterminate-length quantum codes in 1995 by Schumacher \cite{Schumacher95}. 
Schumacher showed that, for any $\delta>0$, 
$(S(\rho)+\delta)$-many qubits per symbol are sufficient in order to encode strings of symbol states if each entry of the sequence is prepared in a 
i.i.d.\ way, with probability of error tending to zero as the length of the strings tends to infinity. 
Here $\rho=\rho_{\mc S}$ is the ensemble state representing the quantum ensemble $\mc S$, and $S(\rho)$ is the von~Neumann entropy of the density matrix $\rho$ given by 
$$
S(\rho)= -\text{tr}\, ( \rho \log_2\rho ) .
$$
Moreover, Schumacher showed that for any $R<S(\rho)$, if at most $R$ qubits are used per symbol, then the probability of error tends to 1 as the length of the strings tends to infinity. 
Thus the von Neumann entropy $S(\rho)$ can be interpreted as the 
minimum expected number of qubits per symbol that are necessary in order to encode strings of symbol states with arbitrarily small error (i.e.\ asymptotically lossless coding) given that the elements of the string of symbol states are prepared in an i.i.d.\ way.

Indeterminate-length quantum codes were considered by Schumacher and Westmoreland in \cite{SW01}, and later by M\"uller, Rogers and Nagarajan in \cite{MR08, MRN09}; and Bellomo, Bosyk, Holik and Zozor in \cite{BBHZ17}. 
In all three of these papers, the authors prove a version of the quantum Kraft-McMillan Theorem which states that every uniquely decodable quantum code must satisfy an inequality in terms of the lengths of its eigenstates. 
Their presentations are very similar to that of the classical 
Kraft-McMillan Theorem (\cite[Theorems 5.2.1 and 5.5.1]{CT91}) except that these authors did not provide a converse statement. 
In Theorem~\ref{quantum Kraft}, we present a modified version of the quantum Kraft-McMillan Theorem giving a converse statement, thus characterizing the uniquely decodable quantum codes. 
Our Theorem~\ref{quantum Kraft} comes in handy when we define an optimal quantum code that corresponds to a given ensemble. 

In Subsection~\ref{Q Data Compression} we introduce the 
notion of quantum stochastic ensemble and Markov ensemble, 
allowing us to prepare strings of symbol states for 
quantum data compression such that the appearance of each 
symbol in the string may depend on the previous symbols; 
i.e.\ the strings of symbol states are not necessarily 
prepared in an i.i.d.\ way. 
Quantum sources that emit sequences of quantum symbols 
that are not necessarily statistically independent have 
been considered in the literature \cite{KingL98} and they 
are well suited for quantum communications.
A stochastic ensemble is a sequence $(\mathcal{S}^k)_{k \in \mathbb{N}}$, where $\mathcal{S}^k = \{ p(n_1,\ldots, n_k) , |s_{n_1}\cdots s_{n_k} \rangle \}_{n_1,\ldots, n_k=1}^N$ for each $k\in\N$ such that $p$ is the probability mass function of a discrete stochastic process $\X$, $\{|s_n\rangle\}_{n=1}^N$ is a collection of vector states referred to as the symbol states and $p(n_1,\ldots, n_k)$ is the probability that the string of quantum symbols $|s_{n_1}\cdots s_{n_k}\rangle$ is encoded, for each $k\in\N$ and $n_1,\ldots, n_k\in\{1,\ldots, N\}$.   

Our main results, Theorems~\ref{main theorem} and \ref{main theorem 2}, give quantum dynamical entropy interpretations for the average minimum codeword length per symbol as the length of strings of symbol states tend to infinity when the coding is assumed to be lossless. 
These results extend the result of Schumacher \cite{Schumacher95} and Bellomo et al.\ \cite{BBHZ17} which state that for an i.i.d.\ prepared quantum ensemble the optimal codeword length per symbol is equal to the von Neumann entropy of the initial ensemble state for asymptotically lossless coding. 
In our result we use the quantum Markov chain (QMC) approach to quantum dynamical entropy which we recall in Subsection~\ref{sect aow}. The notion of QMC was introduced by Accardi in \cite{Accardi75} and its use for describing dynamical entropy was first appeared in \cite{AOW97} in terms of the Accardi-Ohya-Watanabe (AOW) entropy. Another QMC approach was introduced by Tuyls in \cite{Tuyls98} for the study of the Alicki-Fannes (AF) entropy, which was introduced in \cite{AF94} and often referred to as ALF entropy to emphasize Lindblad's contributions. Finally, a generalization of both QMC approaches was given in \cite{KOW99}, where the authors introduced the Kossakowski-Ohya-Watanabe (KOW) entropy. Throughout this article, we will follow mainly the terminology and notations of \cite{AOW97} and \cite{KOW99}.

\section{Data Compression}\label{applications}

In what follows, all codings will be done into strings of bits or strings of qubits for classical and quantum codes, respectively. Therefore all codewords will be strings of elements from a binary alphabet $A=\{0,1\}$ (in the classical case) or, possibly the superposition of, strings from a quantum binary alphabet $\mc A=\{|0\rangle, |1\rangle\}$ which is an orthonormal basis of the Hilbert space $H_{\mc A}=\C^2$ (in the quantum case). The extensions to $d$-bits or $d$-qubits can easily be done in both cases. 

\subsection{Classical Codes and the Kraft Inequality}

Let $S$ be a finite or countable set equipped with the power set $\sigma$-algebra $\mc P(S)$, and let $X$ be a random variable with values in $S$. 
The set $S$ will be referred to as the \textbf{symbol set} that we wish to encode. 
In the literature, the set $S$ is referred to as the set of objects, the message set, or sometimes even the index set. 
For any set $Y$, we will set $Y^+$ equal to the set $\cup_{\ell=0}^\infty Y^{\ell}$ which is the collection of all possible finite strings from $Y$, where $Y^0$ denotes the empty set (or empty string). 
Lastly, let $A=\{0,1\}$ be the binary alphabet. 
A \textbf{code} $C:S\ra A^+$ is a mapping from $S$ to $A^+$, the set of finite strings with letters in the binary alphabet $A$. 
The range of the code, $A^+$, is referred to as the \textbf{codebook} and its elements are the \textbf{codewords}. 
Moreover, for each $x\in S$, we refer to $C(x)$ as the \textbf{codeword of the symbol \bma{$x$}}. 
For each $a\in A^+$, we call the \textbf{length of \bma{$a$}} (denoted by $\ell(a)$) the unique integer $m$ such that $a\in A^m$.

The \textbf{expected length} of a code $C$ on a symbol set $S$ is given by 
\begin{equation*}\label{expected length}
EL(C):= \sum_{x\in S}p(x)\ell(C(x))=\E[\ell (C (X))],
\end{equation*}

\noindent where $p:S\ra [0,1]$ is the probability mass function (pmf) of the random variable $X$ and the expectation $\E$ is taken with respect to $p$.

We extend the code $C$ by concatenation to obtain the \textbf{extended code}, also called the extension of $C$, $C^+: S^+\ra A^+$. That is to say 
\begin{equation*}\label{extended code}
C^+(x_1x_2\cdots x_n)=C(x_1)C(x_2)\cdots C(x_n)\quad \text{for all }x_1x_2\cdots x_n\in S^n\text{ and } n\in\N, 
\end{equation*}

\noindent and we define $C^+(\emptyset)=\emptyset$. 
We call the code $C$ \textbf{uniquely decodable} whenever its extension $C^+$ is injective; i.e.\ $C$ is uniquely decodable whenever all strings of symbols from $S$ are pairwise distinguishable. 
In lossless coding we are only interested in uniquely decodable codes. 

An extremely useful class of uniquely decodable codes are the so-called \textbf{instantaneous (or prefix-free)} codes. A code is said to be prefix-free if no codeword is the prefix of another; i.e.\ for every distinct pair $x,y\in S$ there is no $a\in A^+$ such that $C(x)a=C(y)$. Prefix-free codes are called instantaneous because the decoder is able to read out each codeword from a string of codewords, instantaneously, as soon as she sees that word appear in a string (without waiting for the entire string). 

The Kraft-McMillan Inequality is fundamental in classical data compression.

\begin{Thm}(Kraft-McMillan Inequality, \text{\cite[Theorems 5.2.1 and 5.5.1]{CT91}})\label{Kraft}
For any uniquely decodable code over a symbol set $S$ with cardinality $|S|=m\in \N$, the codeword lengths $\ell_1,\ell_2,\ldots, \ell_m$ must satisfy the inequality $$\sum_{i=1}^m 2^{-\ell_i}\le 1.$$
Conversely, given a set of codeword lengths that satisfies this inequality, there exists an instantaneous code with these code lengths. 
\end{Thm}

\begin{Rem}
The Kraft-McMillan Inequality is sometimes referred to only as the Kraft Inequality. 
This is due to the fact that Kraft was the first to prove the inequality in \cite{Kraft49}, although his original result refers only to instantaneous codes. McMillan later extended Kraft's work to include all uniquely decodable codes in \cite{Mcmillan56}. 
Furthermore, it is worth noting that the Kraft-McMillan inequality can be extended to a countable set of symbols (see Theorem~5.2.2 and the corollary following Theorem~5.5.1 in \cite{CT91}). When including countable sets of symbols, the inequality is referred to as the Extended Kraft-McMillan Inequality. 
\end{Rem}

An immediate corollary to the Kraft-McMillan Inequality is the following:

\begin{Cor}
Given any uniquely decodable code with codeword lengths $\ell_1,\ell_2,\ldots, \ell_m$, there exists an instantaneous code with these same code lengths. 
\end{Cor}

We call a uniquely decodable code $C$ \textbf{optimal} whenever the expected length $EL(C)$ is minimized; i.e.\ the optimal uniquely decodable code is given by 
\begin{equation}\label{optimality}
\begin{split}
C_{\text{opt}}:&= \argmin_C \{EL(C) : C\text{ is uniquely decodable}\} \\
&= \argmin_C \{EL(C) : \text{the codeword lengths of the }C\text{ satisfy} \sum_i 2^{-\ell_i}\le 1\},
\end{split}
\end{equation}
where the last equality follows from Theorem~\ref{Kraft}. 
We set $EL^*(X):=EL(C_{\text{opt}})$ the \textbf{optimal expected length of the random variable \bma{$X$}}. 
The results for the optimal expected length are summarized in the following: 

\begin{Thm}(\text{\cite[Theorem 5.4.1]{CT91}})\label{optimal code}
Let $X$ be a random variable with range in the symbol set $S$. 
Then the optimal expected length of $X$ satisfies the inequality 
$$H(X)\le EL^*(X)<H(X)+1,$$ 
where $H(X)$ is the Shannon entropy of $X$, i.e.\ $H(X)=-\sum_{i\in S} p_i\log_2 p_i$ where $(p_i)_{i\in S}$ is the pmf of $X$. 
\end{Thm}

\noindent Well known examples of codes which satisfy the inequality of Theorem~\ref{optimal code} are the so-called Huffman codes and Shannon-Fano codes. 

In the above theorem, we are only interested in the compressability of single codewords. Suppose instead that we wish to compress strings of codewords with code distributions given by a stochastic process $\X=(X_i)_{i=1}^\infty$. 
Then, for each $n\in \N$, Theorem~\ref{optimal code} holds for the random vector $(X_1,X_2,\ldots, X_n)$, giving 
$$ H(X_1,X_2,\ldots, X_n)\le EL^*(X_1,X_2,\ldots,X_n)< H(X_1,X_2,\ldots, X_n)+1.$$ 
For each $n\in\N$, we set 
\begin{equation}\label{expected length n}
EL^*_n(\X):=\frac{1}{n}EL^*(X_1,X_2,\ldots, X_n)
\end{equation} 
to be the \textbf{optimal expected codeword length per symbol} for the first $n$ symbols. 
We can then express the optimal expected codeword length per symbol (over all symbols) in terms of the entropy rate, which is a dynamical entropy for stochastic processes. The \textbf{entropy rate} of a stochastic process $\X=(X_n)_{n=1}^\infty$ is given by 
\begin{equation*}
H(\X)=\lim_{n\ra\infty} \frac{1}{n} H(X_1,\ldots, X_n), 
\end{equation*} 
whenever the limit exists. There are many instances when it is known that the above limit exists (e.g.\ stationary stochastic processes, see \cite[Theorem 4.2.1]{CT91}). 

\begin{Thm}(\text{\cite[Theorem 5.4.2]{CT91}})\label{expected per symbol}
The optimal expected codeword length per symbol for a stochastic process $\X=(X_i)_{i=1}^\infty$ satisfies 
\begin{equation*}
\frac{H(X_1,X_2,\ldots, X_n)}{n}\le EL^*_n(\X) < \frac{H(X_1,X_2,\ldots, X_n)}{n}+\frac{1}{n}.
\end{equation*}
Moreover, if $\X$ is such that the limit defining entropy rate exists (e.g.\ $\X$ is a stationary stochastic process), then 
\begin{equation*}
EL_n^*(\X)\ra H(\X)\quad\text{as }n\ra\infty.
\end{equation*}
In particular, if $\X$ consists of independent identically distributed (i.i.d.) copies of a random variable $X$, then 
\begin{equation*}
EL_n^*(\X)\ra H(X)\quad\text{as }n\ra\infty.
\end{equation*} 
\end{Thm}

This finishes our brief overview of data compression in classical information theory. For a more detailed exposition see \cite[Chapter 5]{CT91}.

\subsection{Quantum Data Compression}\label{Q Data Compression}

We begin with the description of indeterminate-length quantum codes, whose preliminary investigation began with Schumacher \cite{Schumacher94} and Braunstein et.\ al in \cite{BFGH00}, and they were formalized in \cite{SW01}. 
We may think of the codes introduced in the previous section as being varying-length codes; the term indeterminate-length is used to draw attention to the fact that a quantum code must allow for superpositions of codewords, including those superpositions containing codewords with different lengths. 
We will follow mainly the formalisms in \cite{BBHZ17} as opposed to the zero-extended forms of \cite{SW01}. 
A description of the connection between these two formalisms can be found in \cite{BF02}.

For any Hilbert space $H$, we will denote by $H^\oplus:=\oplus_{\ell=0}^\infty H^{\ten \ell}$ the \textbf{free Fock space} of $H$, where $H^{\ten 0}=\C$. 
We will denote the scalar $1\in H^{\ten 0}$ by $|\emptyset\rangle$ and refer to it as the \textbf{empty string}. 
Let $\mc S=\{p_n,|s_n\rangle\}_{n=1}^N$ be an \textbf{ensemble of pure states}, or simply ensemble, where $p=\{p_n\}_{n=1}^N$ is the pmf of a random variable $X$ and $|s_n\rangle$ is an element of a $d$-dimensional Hilbert space $H_\mc S$, for each $1\le n\le N$, such that $H_{\mc S}=\span\{|s_n\rangle\}_{n=1}^N$. 
The collection $\{|s_n\rangle\}_{n=1}^N$ will be referred to as the \textbf{symbol states} of the ensemble $\mc S$. 
An \textbf{(indeterminate-length) quantum code}, $U$, over a quantum binary alphabet $\mc A:=\{|0\rangle, |1\rangle\}$, which is an orthonormal basis for $H_\mc A=\C^2$, is a linear isometry $U:H_\mc S\ra H_\mc A^\oplus$. The \textbf{extended quantum code of \bma{$U$}} is the linear mapping $U^+:H_\mc S^\oplus\ra H_\mc A^\oplus$ given by 
\begin{equation*}
U^+(|s_1 s_2\cdots s_n\rangle)=U(|s_1\rangle)U(|s_2\rangle)\cdots U(|s_n\rangle),
\end{equation*}

\noindent for all $|s_1 s_2\cdots s_n\rangle \in H_\mc S^{\ten n}$ and $n\in\N$, and we set $U^+(|\emptyset\rangle)=|\emptyset\rangle$, where concatenation is defined according to \cite[Definition 2.3]{MR08} (see also \cite[Section V]{MRN09}). 

The quantum code $U$ is said to be \textbf{uniquely decodable} if the extended quantum code $U^+$ is an isometry. 
Throughout this paper, we will restrict ourselves only to the situation where the range of $U$ is a subset of $H_{\mc A}^{\oplus \lmax}$ for some $\lmax\in \N$; i.e.\ there is a finite upper bound $\lmax$ on the length of all codewords.

\begin{Rem}
The authors of \cite{BF02} allow non-empty strings to map to the empty string. In their paper, the authors send along a classical side channel to give the lengths of the codewords and so that convention is possible. Without the classical side channel (as is the approach in the present paper) allowing non-empty strings to map to the empty string will cause the quantum code to \textit{not} be uniquely decodable.
\end{Rem}

Let $\mc S=\{p_n,|s_n\rangle\}_{n=1}^N$ be an ensemble whose symbol states $\{|s_n\rangle\}_{n=1}^N$ span a Hilbert space $H_{\mc S}$ of dimension $d$. 
Consider a classical uniquely decodable code, $C$, on a symbol set, $S=\{x_i\}_{i=1}^d$, with $d$-many symbols. 
We will construct a corresponding uniquely decodable quantum code, $U$, from $C$ by identifying the classical binary alphabet $A=\{0,1\}$ with the quantum binary alphabet 
$\mc A=\{|0\rangle,|1\rangle\}\subseteq \C^2$ and the symbol set, $S$, with any orthonormal basis $\{|e_i\rangle\}_{i=1}^d$ of $H_\mc S$; this construction is given in \cite{BBHZ17}. 
Fix an orthonormal basis $\{|e_i\rangle\}_{i=1}^d$ of $H_\mc S$ and define the quantum code $U:H_\mc S\ra H_{\mc A}^\oplus$ by the equation
\begin{equation}\label{C-Q source code}
U=\sum_{i=1}^d |C(x_i)\rangle\langle e_i|.
\end{equation}
It is clear that $|C(x_i)\rangle \in H_\mc A^{\ten \ell_i}\subseteq H_\mc A^\oplus$, where $\ell_i$ is the length of $C(x_i)$, and that $\{|C(x_i)\rangle\}_{i=1}^d$ is an orthonormal set, so that $U$ is a linear isometry. Furthermore, since $C$ is uniquely decodable, the map $U^\ell:H_\mc S^{\ten\ell}\ra H_\mc A^\oplus$ defined by the equation 
\begin{equation*}\label{C-Q ext source code1}
U^\ell =\sum_{i_1=1}^d\cdots \sum_{i_\ell=1}^d |C(x_{i_1})C(x_{i_2})\cdots C(x_{i_\ell})\rangle\langle e_{i_1}e_{i_2}\cdots e_{i_\ell}|
\end{equation*}
is a linear isometry for each $\ell\in\N$. 
Since the extended quantum code $U^+:H_\mc S^\oplus\ra H_\mc A^\oplus$ is given by 
\begin{equation*}\label{C-Q ext source code2}
U^+=\sum_{\ell=0}^\infty U^\ell, 
\end{equation*}
we see that $U^+$ is a linear isometry and hence $U$ is uniquely decodable. We will refer to quantum codes constructed from classical ones by Equation~\eqref{C-Q source code} as \textbf{classical-quantum encoding schemes (c-q schemes)}.

\begin{Rem}
Notice that the symbol states $\{|s_n\rangle\}_{n=1}^N$ of the ensemble $\mc S$ are not directly encoded by the $|C(x_i)\rangle$'s unless $N=d$ and there exists a permutation $\sigma$ of $\{1,\ldots, d\}$ such that  $|s_{\sigma(i)}\rangle=|e_i\rangle$ for every $i\in\{1,\ldots, d\}$. In fact $U|s_n\rangle$ need not belong to $H_\mc S^{\ten\ell}$ for any $\ell\in\N$, but can in general be in a superposition of different lengths. (Hence the term indeterminate-length quantum codes.)
\end{Rem} 

The Kraft-McMillan Inequality (Theorem~\ref{Kraft}) was initially extended to the quantum domain in \cite{SW01} and subsequently in \cite{MR08} and \cite{BBHZ17}. Before presenting (a slightly different) Quantum Kraft-McMillan Inequality, we will first introduce the length observable and quantum codes with length eigenstates. The \textbf{length observable} $\Lambda$ acting on $H_\mc A^\oplus$ is given by 
\begin{equation}\label{length observable}
\Lambda:= \sum_{\ell=0}^\infty \ell \Pi_\ell,
\end{equation}
where $\Pi_{\ell}$ is the orthogonal projection onto the subspace $H_\mc A^{\ten\ell}$ of $H_\mc A^\oplus$. 

We say that a quantum code $U: H_\mc S\ra H_{\mc A}^\oplus$ has \textbf{length eigenstates} if $U$ has the form 
\begin{equation}\label{length eigenstates}
U=\sum_{i=1}^d |\psi_i\rangle\langle e_i|.
\end{equation}
for some orthonormal basis $\{|e_i\rangle\}_{i=1}^d$ of $H_\mc S$ and some sequence $\{|\psi_i\rangle\}_{i=1}^d\subseteq H_{\mc A}^+$ such that, for each $1\le i\le d$, $|\psi_i\rangle\in H_{\mc A}^{\ten \ell_i}$ for some $\ell_i\in\N$. 

Note that the $|\psi_i\rangle$'s are orthogonal due to $U$ being a linear isometry. It is easy to see that every c-q scheme is a quantum code with length eigenstates. 
Lastly, for each $\ell\in\N\cup\{0\}$, we will refer to the elements of the set $\{\psi_i: i\in\{1,\ldots,d\},\ \psi_i\in H_{\mc A}^{\ten \ell}\}$ as the \textbf{length \bma{$\ell$} eigenstates of \bma{$U$}} and we will refer to $\{\ell_i\}_{i=1}^d$, where, for each $i=1,\ldots, d$, $\psi_i\in H_{\mc A}^{\ten \ell_i}$, as the \textbf{length eigenvalues of \bma{$U$}}. 

\begin{Rem}
The quantum versions of the Kraft-McMillan Inequality proved in \cite[Section IIC]{SW01} and \cite[Theorem 3.6]{MR08} are more general than the same proved in \cite[Theorem 1]{BBHZ17}, although the formalisms are quite different in all three. 
Our version of the quantum Kraft-McMillan Inequality, presented below, is a generalization of \cite[Theorem 1]{BBHZ17}, but is not quite in the full generality of \cite[Section IIC]{SW01} (in the forward direction) because we only consider uniquely decodable codes (as opposed to the more general notion called \textit{condensable codes} considered in \cite{SW01}). 
However, our version does have a converse statement, similar to the classical Kraft-McMillan Inequality, which is missing from the aforementioned quantum versions. 
\end{Rem}

\begin{Thm}(Quantum Kraft-McMillan Inequality)\label{quantum Kraft}
Any uniquely decodable quantum code $U$ with length eigenstates over a binary alphabet must satisfy the inequality $$\tr(U^\dagger 2^{-\Lambda}U)\le 1.$$ 
Conversely, if $U:H_{\mc S}\ra H_{\mc A}^\oplus$ is a linear isometry with length eigenstates satisfying the above inequality, then there exists a c-q scheme $\widetilde U$ with the same number of length $\ell$ eigenstates for each $\ell\in \N$. 
\end{Thm}

The proof of Theorem~\ref{quantum Kraft} is presented in the Appendix.

We would like to find a quantum code which minimizes the amount of resources required. Unfortunately there are numerous ways to define the length of a codeword for an indeterminate-length quantum code (e.g.\ base length \cite{BF02}, exponential length \cite[Definition 6]{BBHZ17}, etc.). 
Here, we follow \cite[Definition 3]{BBHZ17} and define the \textbf{length of a codeword \bma{$|\omega\rangle$}}, which is a normalized vector in $H_{\mc A}^{\oplus}$ given by $|\omega\rangle=U|s\rangle$ for a unique symbol state $|s\rangle\in\{|s_n\rangle\}_{n=1}^N$, as the expectation with respect to the length observable in Equation~\eqref{length observable}. 
Explicitly, the length of a codeword $|\omega\rangle=U|s\rangle$ will be given by a function $\ell: H_{\mc A}^{\oplus}\ra \R^+$, defined as follows: 
\begin{equation}\label{q expected length}
\ell(|\omega\rangle) := \langle \omega |\Lambda|\omega\rangle=
\langle Us ,\Lambda Us\rangle= \langle s ,U^\dagger \Lambda U s\rangle.
\end{equation}
Whenever $U$ has length eigenstates and is given by Equation~\eqref{length eigenstates}, we see that Equation~\eqref{q expected length} simplifies to 
\begin{equation*}
\ell(|\omega\rangle) = \sum_{i=1}^d |\langle e_i|s\rangle|^2\ell_i,
\end{equation*}
where $\{\ell_i\}_{i=1}^d$ denotes the set of length eigenvalues of $U$. 

Again we follow \cite{BBHZ17} and, for any ensemble $\mc S=\{p_n,|s_n\rangle\}_{n=1}^N$, we define the \textbf{ensemble state} $\rho_{\mc S}$ of $\mc S$ by 
\begin{equation*}\label{ensemble state}
\rho_{\mc S}=\sum_{n=1}^N p_n |s_n\rangle\langle s_n|\in S_1(H_{\mc S}).    
\end{equation*} 
If $U$ is a quantum code on $H_{\mc S}$ define the \textbf{average codeword length} with respect to the ensemble $\mc S$ by 
$$EL(U)=\tr(\rho_{\mc S} U^\dagger \Lambda U).$$ 
We denote by $U_{\text{opt}}$ the \textbf{optimal quantum code with length eigenstates} for the ensemble $\mc S$ if 
\begin{equation*}\label{q optimal code}
\begin{split}
U_{\text{opt}}:&=\argmin_U \{EL(U) :U \text{ is uniquely decodable with length eigenstates}\}\\
&=\argmin_U \{EL(U) :\tr(U^\dagger 2^{-\Lambda}U)\le 1\} \\
&=\argmin_U \{EL(U) :U\text{ is a c-q scheme 
	satisfying }\sum_{i=1}^d 2^{-\ell_i}\le 1\} ,
\end{split}
\end{equation*}
where the second and third equality follow from Theorem~\ref{quantum Kraft}, and the $\{\ell_i\}_{i=1}^d$ in the third equality denote the length eigenvalues of $U$. 
The existence of $U_\opt$ follows from the existence of $C_\opt$ in Equation~\eqref{optimality} by the backward direction of Theorem~\ref{quantum Kraft}. 
The \textbf{optimal average codeword length} for the ensemble $\mc S$ is given by 
\begin{equation}\label{minimum length}
EL^*(\rho_{\mc S}):= EL(U_\opt)= \tr(\rho_{\mc S} U_\opt^\dagger \Lambda U_\opt).
\end{equation}

\noindent It is shown in \cite[Theorem 2]{BBHZ17} that the optimal c-q scheme (and hence optimal quantum code with length eigenstates by the converse of Theorem~\ref{quantum Kraft}) is given by the classical Huffman codes. The bounds on $EL^*(\rho_{\mc S})$ in terms of the von-Neumann entropy follow immediately. 

\begin{Thm}\label{bounds}
The minimum average codeword length for an ensemble $\mc S$ is bounded as follows, 
$$ S(\rho_{\mc S})\le EL^*(\rho_{\mc S}) < S(\rho_{\mc S})+1.$$
\end{Thm}

\begin{proof}
See \cite[Theorem 3]{BBHZ17}. 
\end{proof}

Next, we wish to consider the optimal average codeword length per symbol for a collection of ensembles $\{\mc S^k\}_{k=1}^\infty$, where $\mc S^k=\{p(n_1,\cdots, n_k), |s_1s_2\cdots s_k\rangle \}_{n_1,\ldots,n_k=1}^N$ and probabilities given by the pmf $p$ of a stochastic process $\X$. 
We will refer to such collections of ensembles as \textbf{stochastic ensembles}. 
Note that, by the definitions of a stochastic process, a stochastic ensemble $\mc S^k$ must be compatible in the following sense: 
\begin{equation*}\label{compatible}
\sum_{n_{k+1}=1}^N p(n_1,\cdots,n_{k},n_{k+1})=p(n_1,\cdots, n_{k}),
\end{equation*} 
for all $n_1,\ldots,n_{k}\in\{1,\ldots, N\}\text{ and }k\in\N$.
Notice that we allow for the possibility that preparations of the ensemble at each time be dependent upon previous preparations. If the preparations of the ensemble are independent and identically prepared copies of $\mc S=\{p_n,|s_n\rangle\}_{n=1}^N$; i.e.\ the stochastic process $\X$ is made up of i.i.d.\ copies of a random variable $X$, then $p(n_1,\cdots, n_k)=p_{n_1}p_{n_2}\cdots p_{n_k}$ and $\rho_{\mc S^k}=\rho_{\mc S^1}^{\ten k}$, where $\rho_{\mc S^k}=\sum_{n_1,\ldots,n_k=1}^N p(n_1,\ldots,n_k)|s_{n_1}\cdots s_{n_k}\rangle\langle s_{n_1}\cdots s_{n_k}|$. 
For each $k\in\N$, let 
\begin{equation*}\label{n optimal length}
EL^*_k(\rho_{\mc S^k})=\frac{1}{k}EL^*(\rho_{\mc S^k})
\end{equation*}
be the \textbf{optimal average codeword length per symbol} for the first $k$ symbols with respect to the ensemble $\mc S^k$, where $EL^*(\rho_{\mc S^k})$ is given by Equation~\eqref{minimum length}. 
Notice that the optimal average codeword length per symbol is defined analogously to the classical case in 
Equation~\eqref{expected length n}. Then, from Theorem~\ref{bounds}, we have 
\begin{equation}\label{bounds 2}
\frac{1}{k}S(\rho_{\mc S^k})\le EL^*_k(\rho_{\mc S^k})< \frac{1}{k}S(\rho_{\mc S^k})+\frac{1}{k}.
\end{equation}

\noindent In the following section, we will relate the above quantities to the dynamical entropy of a quantum dynamical system.


\section{An expression for the optimal quantum data compression rate using quantum dynamical entropy}\label{QDC via QDE}

\subsection{A quantum dynamical system associated with a stationary Markov ensemble}\label{OQRW sect}

In this article we consider stochastic processes 
$\X =(X_n)_{n=1}^\infty$ such 
that for some fixed $N < \infty$, each $X_n$ is a random variable with values in $\{1,\ldots, N\}$. If $\X$ is a 
stochastic process set $p_\X$ to be the 
pmf of $\X$, i.e.\ $p_\X$ is a probability measure on 
$\{ 0, \ldots , N\}^\N$, such that for every $k \in \N$ and 
$n_1, \ldots , n_k \in \{ 1, \ldots , N \}$, $p_\X (n_1, \ldots , n_k)= \pr [ X_1 =n_1, \ldots , X_k =n_k]$.  We define the associated stochastic ensemble 
$\{\mc S^k\}_{k=1}^\infty$ by setting 
$\mc S^1=\{ p_{\X}(n), |s_n\rangle\}_{n=1}^N$ whose symbol states span $H_{\mc S}$ 
and $\mc S^k=\{ p_{\X}(n_1,\ldots, n_k), |s_{n_1}\cdots s_{n_k}\rangle \}_{n_1,\ldots, n_k=1}^N$ whose symbol states span 
$H_\mc S^{\ten k}$ 
for each $k\in\N$. 

Recall that a stochastic process $\X$ is called \textbf{stationary} if the measure $p_{\X}$  is invariant with respect to the translation map, i.e.\ if for every $k \in \N$, $p_{\X}=p_{\Y^{(k)}}$
where the stochastic process $\Y^{(k)}=(Y^{(k)}_n)_{n \in \N}$ is defined by 
$Y^{(k)}_n=X_{n+k}$ for every $n \in \N$. Obviously, if the 
stochastic process $\X= (X_n)_{n=1}^\infty$ is stationary 
then the random variables $X_n$ are identically distributed.
The stochastic process $\X$ is called a \textbf{Markov process} if $\pr [X_{k+1}|X_1, \ldots , X_k]= \pr [ X_{k+1}|X_k]$
for every $k \in \N$.
Hence a stochastic process $\X=(X_n)_{n=1}^\infty$ is a \textbf{stationary Markov process} 
if and only if  the random variables $X_n$ are identically distributed and
there exists a column stochastic matrix $(p_{i,j})_{i,j=1}^N$ such that for every $k \in \N$ and
$n_1, \ldots , n_k , n_{k+1} \in \{ 1, \ldots , N\}$, we have 
$$
\pr [X_{k+1}=n_{k+1}|X_1=n_1, \ldots , X_k=n_k]=
\pr[X_{k+1}=n_{k+1}|X_k=n_k]=p_{n_{k+1},n_k}.
$$
Equivalently, a stochastic process $\X=(X_n)_{n=1}^\infty$ is a stationary Markov process if and only if the random 
variables $X_n$ are identically distributed and there exists a transition matrix $P=(p_{n,m})_{n,m=1}^N$ and an initial distribution $p=\{p_{n}\}_{n=1}^N$ such that the pmf of $\X$ is given by $p_{\X}(n_1,\ldots, n_k)=p_{n_1}\prod_{l=2}^k p_{n_l, n_{l-1}}$, for each $k\in\N$ and $1\le n_1,\ldots, n_k\le N$, and $p$ is invariant with respect to $P$; i.e.\ $Pp=p$; (the reason that $p$ is called ``initial distribution" is because it coincides with the distribution of the random variable $X_1$).  
We will refer to a stochastic ensemble governed by a Markov process $\X$ as the \textbf{Markov ensemble governed by \bma{$\X$}}. 
Whenever the Markov process is stationary we will refer to the Markov ensemble as being stationary.

Let $\{\mc S^k\}_{k=1}^\infty$ be a stationary Markov ensemble governed by a stationary Markov process $\X$ having  transition matrix $P=(p_{n,m})$ and initial distribution $p=\{p_n\}_{n=1}^N$. 
Setting $d=\dim(H_{\mc S})$, so that $d^k=\dim(H_{\mc S}^{\ten k})$ for each $k\in\N$, the following sequence of ensemble states which represent this collection of ensembles is defined:  
\begin{equation*}\label{rho1}
\rho_{S^1}= \sum_{n=1}^N p_n |s_n\rangle\langle s_n| \in M_d=S_1(H_{\mc S})
\end{equation*}
and for each $k\in\N$ with $k\ge 2$,
\begin{equation*}\label{rhok}
\begin{split}
\rho_{S^k}&= \sum_{n_1,\ldots, n_k=1}^N p_{n_1}\prod_{l=2}^k p_{n_l, n_{l-1}} |s_{n_1}\cdots s_{n_k}\rangle\langle s_{n_1}\cdots s_{n_k}|  \\
&= \sum_{n_1=1}^N p_{n_1} |s_{n_1}\rangle\langle s_{n_1}|\ten \cdots \ten \sum_{n_k=1}^N p_{n_k, n_{k-1}} |s_{n_k}\rangle\langle s_{n_k}|\in M_d^{\ten k}=S_1(H_{\mc S}^{\ten k}),  
\end{split}
\end{equation*}
where we used the following notation.

\begin{Not}
If $H$ is a Hilbert space then $S_1(H)$ will denote the space of trace-class operators on $H$. 
In the sequel we will frequently donote by $\mc A$ a unital $C^*$-algebra and $\Sigma (\mc A)$ 
will denote the set of normal states on $\mc A$. 
Since we assume that $\mc A = B(H)$, we will identify each normal state, 
$\omega\in\Sigma(\mc A)$ with its density operator 
$\rho\in S_1(H)$ through the identification $\omega(\cdot)=\tr(\rho\ \cdot)$. 
\end{Not}

We now define a quantum dynamical system associated with
the above stationary Markov ensemble. Recall that a 
\textbf{quantum dynamical system} is a triplet 
$(\mc A, \Theta, \rho)$ where $\mc A$ is a unital 
$C^*$-algebra,
 $\Theta : \mc A \to \mc A$ is a positive unital map, and 
$\rho$ is a density operator on $\mc A$. In some situations the map 
$\Theta$ is taken to be a $*$-automorphism, but we will
not adopt this restriction here. The reason that we assume that 
$\Theta$ is positive and unital, is because we would like to
have that the dual map $\Theta^\dagger$ maps the set of 
density operator of $\mc A$ (i.e.\ the set of positive unital functionals on
$\mc A$) to itself.  
Throughout this paper we will, for simplicity, ignore the GNS construction and when we do not specify the $C^*$-algebra 
we will assume that it is equal to 
$\mc A = B(H)$ for some Hilbert space $H$.

The quantum dynamical system associated to the above 
stationary Markov ensemble is defined as follows:
Let $\mc A= B(\C^N)$, $\Theta: B(\C^N) \to B(\C^N)$ 
be defined by 
\begin{equation} \label{E:MarkovTheta}
\Theta (| k \rangle \langle \ell |)= \delta_{k,\ell} \sum_{i=1}^N p_{k , i}
|i \rangle \langle i |, 
\end{equation}
and let the density operator $\rho \in S_1 (\C^N)$ be defined by 
\begin{equation} \label{E:MarkovRho}
\rho = \sum_{n=1}^N p_n |n \rangle \langle n | . 
\end{equation}

It is easy to verify that the map $\Theta$ is 
positive and unital. Indeed,
$$
\Theta (\1_{\C^N}) = \Theta \left(\sum_{k=1}^M | k \rangle \langle k|\right) = \sum_{k,i=1}^N p_{k,i}|i \rangle \langle k| = \sum_{i=1}^N | i \rangle \langle i| \sum_{k=1}^N p_{k,i} = \sum_{i=1}^N | i \rangle \langle i | = \1_{\C^N} ,
$$
thus $\Theta$ is unital. Also it is easy to verify that the dual map 
$\Theta^\dagger :S_1(\C^N) \to S_1 (\C^N)$ is given by
\begin{equation} \label{E:MarkovThetaDagger}
\Theta^\dagger (| m \rangle \langle n |) = 
\delta_{m,n} \sum_{i=1}^N p_{i,n} | i \rangle \langle n | .
\end{equation}
Note that throughout the article  we use dagger to denote trace-class duality i.e.
$$
\tr (\Theta^\dagger (| m \rangle \langle n| ) |k \rangle \langle \ell |) 
= \tr (| m \rangle \langle n| \Theta ( |k \rangle \langle \ell |) ),
$$
for all $m,n,k,\ell \in \{ 1, \ldots , N \}$, and star to denote 
Hermitian conjugate with respect to the underlying Hilbert space, $\C^N$.
We thus obtain that the dual map $\Theta^\dagger$ is 
positive and trace-preserving. 

Finally we recall that if $H$ is a Hilbert space, then an \textbf{operational partition of unity on $H$} is a 
family $\gamma=(\gamma_i)_{i=1}^d$ for some $d \in \N$, 
satisfying 
$\gamma_i\in B(H)$, for each $i$, and  
$\sum_{i=1}^d \gamma_i^* \gamma_i =\1_H$. 
Let 
$\gamma=(\gamma_i)_{i=1}^d$ be the operational partition of unity on the Hilbert space $\C^N$ 
associated with the above Markov ensemble, defined as follows:

\begin{equation}\label{partition}
\gamma_i:= \sum_{n=1}^N \langle e_i | s_n\rangle |n\rangle\langle n|,\quad\text{for all }i=1,\ldots, d,
\end{equation}
where $\{e_i\}_{i=1}^d$ is a fixed orthonormal basis of 
$H_{\mc S}$. 
Notice that 
\begin{equation*}
\begin{split}
\sum_{i=1}^d\gamma_i^*\gamma_i&= \sum_{i=1}^d \sum_{m,n=1}^N\langle e_i | s_m \rangle 
\overline{\langle e_i | s_n \rangle} |n\rangle\langle n| m\rangle\langle m| \\
&=  \sum_{n=1}^N (\sum_{i=1}^d|\langle e_i | s_n\rangle|^2) |n\rangle\langle n| \\
&=\sum_{n=1}^N \|s_n\| |n\rangle\langle n| = \1_{\C^N},
\end{split}
\end{equation*}
where the second to last equality follows by Parseval's identity. 
Hence $\gamma$ is indeed an operational partition of unity.

\subsection{Quantum Dynamical Entropy via Quantum Markov Chains }\label{sect aow}

In this subsection we recall the definition of quantum Markov chains (QMCs) and dynamical entropy thereon. 
Fix a quantum dynamical system $(\mc A,\Theta,\rho)$
with $\mc A =B(H)$ for some Hilbert space $H$
and fix an operational partition of unity 
$\gamma = (\gamma_i)_{i=1}^d$ on $H$. 
Following \cite[Page 413]{Tuyls98} (see also \cite[Equation 3.14]{KOW99}), we will consider the \textbf{transition expectation} 
$$
\mc E_\gamma: M_d\ten \mc A = B( \C^d \ten H) \ra \mc A
$$
given by the equation 
\begin{equation}\label{trans exp 1}
\mc E_{\gamma}([a_{i,j}]_{i,j=1}^d)= \sum_{i,j=1}^d \gamma_i^* a_{i,j} \gamma_j \quad\text{for all }[a_{i,j}]_{i,j=1}^d=\sum_{i,j=1}^d |e_i\rangle\langle e_j|\ten a_{i,j}\in M_d\ten \mc A, 
\end{equation}
for some fixed orthonormal basis $\{e_i\}_{i=1}^d$ of $\C^d$. 
Further we define the transition expectation 
\begin{equation}\label{trans exp 3}
\mc E_{\gamma,\Theta}=\Theta\circ \mc E_\gamma
:M_d \ten \mc A \to \mc A. 
\end{equation}
Its dual map
$$
\mc E_{\gamma , \Theta}^\dagger =\mc E_\gamma^\dagger \circ \Theta^\dagger : 
S_1 (H)= \Sigma (\mc A)  \to S_1 (\C^d \ten H) = \Sigma (M_d \ten \mc A) ,
$$
(which is defined using trace duality), is usually called a lifting
because it ``lifts" states from $\mc A$ to $M_d \ten \mc A$.

If $H$ is a Hilbert space, $\mc A$ is the von~Neumann algebra $B(H)$ of all bounded operators on $H$, $\rho$ is a density operator on $H$ and, for some $d\in\N$, $\mc E: M_d\ten \mc A\ra \mc A$ is a transition expectation, then the pair $\{\rho, \mc E\}$ is called a \textbf{quantum Markov chain} (QMC).  
We will be specifically interested in QMCs whose transition expectation is given by Equation~\eqref{trans exp 3}. Given a quantum Markov chain, we define the \textbf{quantum Markov state} $\psi$ on $M_d^{\ten \N}$ 
by the equation 
\begin{equation}\label{quantum Markov chain}
\psi(a_1\ten\cdots \ten a_n)
= \tr( \rho \mc E(a_{1}\ten  \mc E(a_{2}\ten  \mc E(\cdots  \mc E(a_{n}\ten\1_H )\cdots)))),
\end{equation}
for all $n\in\N$ and $a_1,\ldots, a_n\in M_d$. Notice that the assumption that the transition expectation $\mc E$ is 
unital implies that $\psi$ is compatible in the sense that 
$$ \psi(a_1\ten\cdots\ten a_n\ten \1_{\C^d})=\psi(a_1\ten\cdots\ten a_n),$$
for all $n\in\N$ and $a_1,\ldots, a_n\in M_d$. Moreover, it was shown in \cite[Proposition 3.7]{Accardi76} that the state $\psi$ on $M_d^{\ten \N}$ indeed exists.

The \textbf{joint correlations} for $\psi$ are given by the density matrices $\rho_n\in M_d^{\ten n}$ satisfying 
\begin{equation}\label{joint correlations}
\psi(a_1\ten\cdots \ten a_n)= \tr(\rho_n a_1\ten \cdots\ten a_n),
\end{equation}
for all $n\in\N$ and $a_1,\ldots, a_n\in M_d$. 

Putting the above pieces together, if $\Theta:\mc A\ra \mc A$ 
is a positive, unital map on the von~Neumann algebra 
$\mc A = B(H)$, $\rho$ is a density operator on 
$H$, and $\gamma = (\gamma_i)_{i=1}^d$ is an operational
partition of unity of $H$, then the 
\textbf{dynamical entropy of \bma{$(\mc A,\Theta,\rho)$} with respect to \bma{$\gamma$}} is given by 
\begin{equation}\label{dyn entropy 1}
h(\Theta, \rho, \gamma)=\limsup_{n\ra\infty}\frac{1}{n}S(\rho_n),
\end{equation}
where $S(\cdot)$ is the von-Neumann entropy and the transition expectation is given by Equation~\eqref{trans exp 3}.
Further, given a subalgebra $\mc B$ of $\mc A$, the \textbf{dynamical entropy of} $(\mc A,\Theta,\rho)$ with respect to $\mc B$ is given by 
\begin{equation*}\label{dyn entropy 2}
h_{\mc B}(\Theta,\rho)=\sup_{\gamma\subseteq \mc B} h(\Theta,\rho, \gamma).
\end{equation*}

\begin{Rem}
The dynamical entropy above is the generalized AF dynamical entropy as defined by the authors of \cite{KOW99}. 
The description we give is very similar to that of the AF dynamical entropy given by Tuyls in \cite{Tuyls98}; however, we do not restrict ourselves to $*$-automorphisms as does the standard construction of AF dynamical entropy. 
\end{Rem}

\subsection{Computation of the quantum dynamical entropy of the quantum dynamical system defined in Subsection~\ref{OQRW sect}}

Let $(\mc A, \Theta , \rho )$ be the quantum dynamical system 
defined by (\ref{E:MarkovTheta}) and (\ref{E:MarkovRho}), and
let $\gamma$ be the operational partition of unity defined
by (\ref{partition}). In this subsection, we will use the 
definitions given in Subsection~\ref{sect aow} in order to
 compute the quantum dynamical
entropy $h(\Theta , \rho , \gamma)$ and give its interpretation as the optimal compression rate of the quantum Markov ensemble that we coinsider.

First we define vectors 
$$
|s_n'\rangle := | s_n \rangle \otimes |n \rangle \in H_{\mc S} \ten 
\C^N \quad \text{for }n =1, \ldots , N,
$$
(which are orthonormal even though the vectors 
$(| s_n \rangle )_{n=1}^N \subseteq H_{\mc S}$ are not necessarily mutually orthogonal), and the state
\begin{equation}\label{qc state}
\rho ' := \sum_{n=1}^N p_n|s_n'\rangle\langle s_n'| =\sum_{n=1}^N p_n |s_n\rangle\langle s_n|\ten |n\rangle\langle n|
\in M_d\ten M_N=S_1(H_{\mc S}\ten \C^N).
\end{equation}
Before proceeding with the construction of the quantum Markov chain, we give a technical lemma which will be helpful later. 

\begin{Lem}\label{lemma1}
Let $\{\mc S^k\}_{k=1}^\infty$ be a stationary Markov ensemble, with symbol states $\{|s_n\rangle\}_{n=1}^N$, which is governed by a stationary Markov process $\X$ with transition matrix $P=(p_{n,m})$. 
Let $\Theta$, $\gamma$, $\mc E_{\gamma}$ and $\mc E_{\gamma,\Theta}$ be defined as above. Then the lifting 
$\mc E_{\gamma,\Theta}^\dagger: S_1(\C^N)\ra 
S_1(H_{\mc S}) \ten S_1(\C^N)$ acts on the diagonal states of $S_1(\C^N)$ in the following way.
$$
\mc E_{\gamma,\Theta}^\dagger (|n\rangle\langle n|) = \sum_{m=1}^N p_{m,n}  |s_m'\rangle\langle s_m'|,
$$
for each $|n\rangle$ in the orthonormal basis of $\C^N$. Moreover, $$\mc E_{\gamma,\Theta}^\dagger (\rho)=\rho'$$
where $\rho$ and $\rho'$ are given by 
Equations~\eqref{E:MarkovRho} and \eqref{qc state}, respectively.
\end{Lem}

The proof of Lemma~\ref{lemma1} can be found in the Appendix.

Next, we will consider the quantum Markov state $\psi$ given by the chain $\{\rho, \mc E_{\gamma,\Theta}\}$, where 
$\rho$ is given in Equation~\eqref{E:MarkovRho} and 
$\mc E_{\gamma,\Theta}$ is as in 
Equation~\eqref{trans exp 3}.  
Then, for each $k\in\N$ and 
$a_1,\ldots, a_k\in B(H_{\mc S})=M_d$, we have 
\begin{equation*}
\begin{split}
\psi(a_1\ten \cdots\ten a_k)&=\tr(\rho \mc E_{\gamma,\Theta}(a_1\ten \mc E_{\gamma,\Theta}(\cdots \mc E_{\gamma,\Theta}(a_k\ten \1_{\C^N})))) 
\quad\text{by Equation~\eqref{quantum Markov chain}} \\
&= \tr(\mc E_{\gamma, \Theta}^\dagger(\rho) a_1\ten \mc E_{\gamma,\Theta}(\cdots \mc E_{\gamma,\Theta}(a_k\ten \1_{\C^N}))) \\
&= \tr(\sum_{n_1=1}^Np_{n_1} |s_{n_1}'\rangle\langle s_{n_1}'| a_1\ten \mc E_{\gamma,\Theta}(a_2\ten \mc E_{\gamma,\Theta}(\cdots \mc E_{\gamma,\Theta}(a_k\ten \1_{\C^N}))))  \\
&= \sum_{n_1=1}^Np_{n_1}\tr(|s_{n_1}\rangle\langle s_{n_1}|a_1)\times \\
& \hspace{1in} \tr(|n_1\rangle\langle n_1|\mc E_{\gamma,\Theta}(a_2\ten \mc E_{\gamma,\Theta}(\cdots \mc E_{\gamma,\Theta}(a_k\ten \1_{\C^N}))))  \\ 
&= \sum_{n_1,n_2=1}^Np_{n_1}p_{n_2,n_1}\tr(|s_{n_1}\rangle\langle s_{n_1}|a_1)\tr(|s_{n_2}\rangle\langle s_{n_2}|a_2)\times \\ 
& \hspace{1in} \tr(|n_2\rangle\langle n_2|\mc E_{\gamma,\Theta}(a_3\ten \mc E_{\gamma,\Theta}(\cdots \mc E_{\gamma,\Theta}(a_k\ten \1_{\C^N})))) \\ 
& \quad \vdots \\
&= \sum_{n_1,\ldots,n_k=1}^Np_{n_1}\prod_{l=2}^k p_{n_l,n_{l-1}}\tr(|s_{n_1}\rangle\langle s_{n_1}|a_1)\cdots \tr(|s_{n_k}\rangle\langle s_{n_k}|a_k),
\end{split}
\end{equation*}
where the ``moreover" part of Lemma~\ref{lemma1} was used in the $3^{\text{rd}}$ equality, the fact $\tr(A\ten B)=\tr(A)\tr(B)$ was used in the $4^{\text{th}}$ equality and Lemma~\ref{lemma1} was used in the $5^{\text{th}}$ equality.

Thus, for each $k\in\N$, the density matrix $\rho_k$ which is defined by Equation~\eqref{joint correlations} is given by 
\begin{equation}\label{rhon=rhok}
\rho_k= \sum_{n_1,\ldots,n_k=1}^Np_{n_1}\prod_{l=2}^k p_{n_l,n_{l-1}} |s_{n_1}\cdots s_{n_k}\rangle \langle s_{n_1}\cdots s_{n_k}| = \rho_{S^k}.
\end{equation} 
Therefore, 
$$
h(\Theta, \rho, \gamma) =\limsup_{k\ra\infty}\frac{1}{k}S(\rho_k) =\limsup_{k\ra\infty}\frac{1}{k}S(\rho_{S^k})=\limsup_{k\ra\infty} EL_k^*(\rho_{\mc S^k}),
$$ 
where the first equality holds by the definition of the dynamical entropy in Equation~\eqref{dyn entropy 1} and the last equality follows from Equation~\eqref{bounds 2}. 
We have proved the following result. 

\begin{Thm}\label{main theorem}
Given any stationary Markov ensemble 
$\{\mc S^k\}_{k=1}^\infty$, the optimal average codeword
 length per symbol (via lossless coding) converges to the dynamical entropy of the above-described quantum dynamical system $(B(\C^N)),\Theta,\rho)$ with respect to the 
 operational partition of unity $\gamma$ defined in 
 Equation~\eqref{partition} in the following sense: 
$$ 
\limsup_{k\ra\infty} EL^*_k(\rho_{S^k})=\limsup_{k\ra\infty}\frac{1}{k}S(\rho_{\mc S^k})=h(\Theta^*,\rho_0, \gamma).
$$
\end{Thm}

We recover the result of Schumacher \cite{Schumacher95} and Bellomo et.\ al \cite{BBHZ17} which states that the optimal codeword length per symbol for an i.i.d.\ prepared ensemble, 
$\{\mc S^k\}_{k=1}^\infty$, (via asymptotically lossless coding)
 is equal to the von Neumann entropy of the initial ensemble
  state, $\rho_{\mc S^1}$: 

\begin{Cor}\label{cor1}
Given a Markov process $\X$ made up of i.i.d.\ copies of a random variable $X$, the stationary Markov ensemble 
$\{\mc S^k\}_{k=1}^\infty$ governed by $\X$ has optimal codeword length per symbol (via lossless coding) given by  
$$
\lim_{k\ra\infty} EL^*_k(\rho_{\mc S^k})=S(\rho_{\mc S^1}).
$$
\end{Cor}

\begin{proof}
First notice that $\X$ is governed by the transition matrix $P=(p_{n,m})_{n,m=1}^N$ such that $p_{n,m}=p_n$, for every $1\le n,m\le N$, where $p=(p_n)_{n=1}^N$ is the initial distribution of $\X$. Therefore 
$$
\rho_{\mc S^k}=\rho_{\mc S^1}^{\ten k}, \quad \text{for each } k\in\N. 
$$
Using the construction from above and Equation~\eqref{rhon=rhok}, we have that 
\begin{equation*}
S(\rho_k)=S(\rho_{\mc S^k})=S(\rho_{\mc S^1}^{\ten k})=kS(\rho_{\mc S^1}),
\end{equation*}
where the last inequality follows by additivity of von Neumann entropy (see e.g.\ \cite[Equation 2.8]{Werhl78}). 
Therefore, by Theorem~\ref{main theorem}, we have 
$$\lim_{k\ra\infty} EL^*_k(\rho_{S^k})=\lim_{k\ra\infty} \frac{1}{k}S(\rho_{\mc S^k})=\lim_{k\ra\infty} \frac{1}{k} kS(\rho_{\mc S^1})=S(\rho_{\mc S^1}).$$
\end{proof}

Next we turn to a similar representation for general stochastic ensembles. We chose to present the case of the stationary Markov ensemble separately since the construction is simpler
than in the general case.


\subsection{A quantum dynamical system associated with a general stochastic ensemble}\label{OQRW sect2}

Consider a stochastic process $\X=(X_n)_{n=1}^\infty$ with values in $\{1,\ldots, N\}$ for some $N<\infty$ and with pmf 
$p$ i.e.\ for any $k \in \N$ and any $(n_1, \ldots , n_k) \in \{ 1, \ldots , N \}^k$ we have $p(n_1, \ldots , n_k)= \pr [X_1=n_1, \ldots , X_k=n_k]$.
Define the associated stochastic ensemble $\{\mc S^k\}_{k=1}^\infty$ by $\mc S^1=\{p(n), |s_n\rangle\}_{n=1}^N$ whose symbol states span $H_{\mc S}$ 
and $\mc S^k=\{p(n_1,\ldots, n_k), |s_{n_1}\cdots s_{n_k}\rangle \}_{n_1,\ldots, n_k=1}^N$ whose symbol states span $H_\mc S^{\ten k}$ 
for each $k\in\N$. 
Again, setting $d=\dim(H_{\mc S})$, so that $d^k=\dim(H_{\mc S}^{\ten k})$ for each $k\in\N$, we define the following sequence of ensemble states which represents this stochastic ensemble:  
\begin{equation*}\label{stochastic rho1}
\rho_{S^1}= \sum_{n=1}^N p(n) |s_n\rangle\langle s_n| \in M_d=S_1(H_{\mc S})
\end{equation*}
and, for each $k\in\N$ with $k\ge 2$, define 
$\rho_{S^k} \in M_d^{\ten k}=S_1(H_{\mc S}^{\ten k}) $ by
\begin{equation*}\label{stochastic rhok}
\begin{split}
\rho_{S^k}&= \sum_{n_1,\ldots, n_k=1}^N p(n_1,\ldots, n_k) |s_{n_1}\cdots s_{n_k}\rangle\langle s_{n_1}\cdots s_{n_k}|  \\
&= \sum_{n_1=1}^N p(n_1) |s_{n_1}\rangle\langle s_{n_1}|\ten \cdots \ten \sum_{n_k=1}^N p(n_k |n_1,\ldots, n_{k-1}) |s_{n_k}\rangle\langle s_{n_k}|. 
\end{split}
\end{equation*}

We define a quantum dynamical system associated to the above quantum ensemble as follows: Let $H= (\C^N)^\oplus= \oplus_{n=0}^\infty (\C^N)^{\otimes n}$ be the free Fock space
 of $\C^N$. Recall that $(\C^N)^{\otimes 0} =\C$ and we denote by 
$|\emptyset \rangle$ the vector $1 \in (\C^N)^{\otimes 0}$. 
We denote $\{ | n \rangle : n \in \{ 1, \ldots , N \} \}$ the standard orthonormal basis of $\C^N$ and 
$$
 \{ | \emptyset \rangle \} \cup \{ |\overline{n} \rangle : 
\bar n \in \{1, \ldots , N \}^k, k \in \N \}
 = \{ | \bar n \rangle : \bar n \in \{ 1, \ldots , N \}^+\}
$$
the standard orthonormal basis of $H$. Before proceeding further we introduce two useful notations on the standard
orthonormal basis of $H$ which will be used later. 
If $\overline{n}=(n_1, \ldots , n_k) \in \{ 1, \ldots , N \}^k$ for some $k \in \N$, then we set
$$
\final (\overline{n})=n_k , \quad 
\pruned (\overline{n})= | n_1 , \ldots , n_{k-1} \rangle \text{ if }k \geq 2 \text{ and } \pruned (\overline{n}) = |\emptyset \rangle
\text{ if }k=1.
$$

Set  $\mc A$ to denote $C^*$-subalgebra of $B(H)$ generated 
by the identity operator and the 
rank-one operators of the form 
$| \bar n \rangle \langle \bar m |$
where  
$\bar n , \bar m \in \{ 1, \ldots , N\}^+$.
 Define a unital map 
$\Theta : \mc A \to \mc A$ by   
$$
\Theta (| \bar n \rangle \langle \bar m |) = 0 
\text{ if at least one of }\bar n , \bar m  \text{ is equal to }\emptyset ,
$$
$$
\Theta (|\bar n \rangle \langle \bar m |) = 0 
\text{ if } \bar n  \not = \bar m ,
$$
$$
\Theta (|\bar n \rangle \langle \bar n |) =
p(\final (\bar n )|\pruned (\bar n )) 
| \pruned (\bar n ) \rangle \langle \pruned (\bar n )| .
$$
It is easy to see that the map $\Theta$ is positive and 
(by definition) unital.
Finally we define a state 
$$
\rho = | \emptyset \rangle \langle \emptyset | ,
$$
on $\mc A$ and consider the dynamical system 
$( \mc A, \Theta , \rho )$.

We also define an operational partition of unity of the Hilbert space $H$ which is associated to the above quantum ensemble. Let  $\{e_i\}_{i=1}^d$ be a fixed orthonormal basis 
of $H_{\mc S}$, and let $\gamma = (\gamma_i)_{i=1}^d$ be
the operational partition of unity defined by
 by 
\begin{equation*}\label{stochastic partition 2}
\gamma_i(|\bar n\rangle)= 
\langle e_i | s_{\final (\bar{n})}\rangle |\bar n\rangle
\text{ for } \bar n \in \{1,\ldots, N\}^k, \quad 
\gamma_i (| \emptyset \rangle ) = \langle e_i | s_1 \rangle |\emptyset \rangle ,
\end{equation*}
for each  $i\in\{1,\ldots , d\}$. Notice that
\begin{equation*}
\sum_{i=1}^d \gamma_i^*\gamma_i (|\bar n\rangle) = \sum_{i=1}^d |\langle e_i, s_{\final (\bar n )}\rangle|^2 |\bar n\rangle 
= \|s_{n_{k-1}}\|^2 |\bar n\rangle = |\bar n\rangle 
\end{equation*}
for $\bar n \in \cup_{k \in \N} \{ 1, \ldots , N \}^k$. 
Similarly,
\begin{equation*}
\sum_{i=1}^d \gamma_i^*\gamma_i (| \emptyset \rangle) = \sum_{i=1}^d |\langle e_i, s_1 \rangle|^2 |\bar n\rangle 
= \|s_1 \|^2 | \emptyset \rangle = |\emptyset \rangle 
\end{equation*}
and thus $\gamma$ is indeed an operational partition of unity on $H$. 
In order to unify the last three displayed formulas, we define
$\final (\emptyset )=1$ and thus we can write
\begin{equation}\label{E:GeneralGamma}
\gamma_i (|\bar n \rangle )= \langle e_i | s_{\final (\bar n )} \rangle | \bar n \rangle \quad \text{for all }\bar n \in \{ 1, \ldots , N \}^+ .
\end{equation} 
Let $\mc E_\gamma$ and $\mc E_{\gamma , \Theta}$ be 
the transition expectation maps from $M_d\ten \mc A$ to 
$ \mc A$ by 
Equation~\eqref{trans exp 1} and \eqref{trans exp 3}, respectively.

Before stating the next result we introduce some notation.

\begin{Not}
For each $k\in\N$ and $\bar n \in\{1,\ldots,N\}^k$ we set 
$$
|s_{\bar n}'\rangle = |s_{\final (\bar n )}\rangle \ten |\bar n\rangle \in H_{\mc S}\ten (\C^N)^{\ten k} \text{ and } 
|s_{\emptyset}'\rangle = |s_1 \rangle\ten |\emptyset\rangle \in H_{\mc S}\ten (\C^N)^{\ten 0} .
$$
Also for $\bar n = (n_1, \ldots , n_k) \in \{ 1, \ldots , N \}^k$ and $\ell \in \{ 1, \ldots , N \}$ we set
$$
\bar n \circ\ell = (n_1, \ldots , n_k , \ell ) \text{ and } \emptyset \circ \ell = \ell .
$$
\end{Not}

We now state a technical lemma which will be used in the 
proof of the main result.

\begin{Lem}\label{lemma2}
	Let $\{\mc S^k\}_{k=1}^\infty$ be a stochastic ensemble
	 with symbol states $\{|s_n\rangle\}_{n=1}^N$ which is governed by a stochastic process $\X$ with pmf $p$. 
	Let $H$, $\mc A$, $\Theta$, $\rho$, $\gamma$, $\mc E_{\gamma}$ and $\mc E_{\gamma,\Theta}$ be defined as above. 
	Then the lifting $\mc E_{\gamma,\Theta}^\dagger: 
	\Sigma (\mc A)\ra M_d \ten \Sigma (\mc A)$ acts on the diagonal states of $S_1(H)$ in the following way: 
$$
\mc E_{\gamma,\Theta}^\dagger (|\bar n\rangle\langle \bar n|) = \sum_{k =1}^N p(k |\bar n)  |s_{\bar n\circ k}'\rangle\langle s_{\bar n\circ k}'|,
$$
for each $|\bar n\rangle $ in the standard orthonormal basis of $H$ where we adopt the convention $p(k|\emptyset) :=p(k)=\pr [X_1=k ]$
for $k \in \{ 1, \ldots , N \}$. 
Moreover, 
$$
\mc E_{\gamma,\Theta}^\dagger (\rho)=\sum_{k=1}^N p(k) |s_{k}'\rangle\langle s_{k}'|.
$$ 
\end{Lem}

The proof of Lemma~\ref{lemma2} can be found in the Appendix.

Next, we will consider the quantum Markov state $\psi$ given by the chain $\{\rho, \mc E_{\gamma,\Theta}\}$. For each 
$k\in\N$ and $a_1,\ldots, a_k\in B(H_{\mc S})=M_d$, we have 
\begin{eqnarray*}
\psi(a_1 \ten &\cdots&\ten a_k)  \\
&=& \tr(\rho \mc E_{\gamma,\Theta}(a_1\ten \mc E_{\gamma,\Theta}(\cdots \mc E_{\gamma,\Theta}(a_k\ten \1_{H})))) 
\quad\text{by Equation~\eqref{quantum Markov chain}} \\
&=& \tr(\mc E_{\gamma, \Theta}^\dagger (\rho) a_1\ten \mc E_{\gamma,\Theta}(\cdots \mc E_{\gamma,\Theta}(a_k\ten \1_{H}))) \\
&=& \tr(\sum_{n_1=1}^Np(n_1) |s_{n_1}'\rangle\langle s_{n_1}'| a_1\ten \mc E_{\gamma,\Theta}(a_2\ten \mc E_{\gamma,\Theta}(\cdots \mc E_{\gamma,\Theta}(a_k\ten \1_{H}))))  \\
&=& \sum_{n_1=1}^Np(n_1)\tr(|s_{n_1}\rangle\langle s_{n_1}|a_1)\tr(|n_1\rangle\langle n_1|\mc E_{\gamma,\Theta}(a_2\ten \mc E_{\gamma,\Theta}(\cdots \mc E_{\gamma,\Theta}(a_k\ten \1_{H}))))  \\ 
&=& \sum_{n_1, n_2=1}^N p(n_1) p(n_2|n_1) \tr(|s_{n_1}\rangle\langle s_{n_1}|a_1)\tr(|s_{n_2}\rangle\langle s_{n_2}|a_2)\times \\ 
&\phantom{a}&  \qquad  \tr(|n_1, n_2\rangle\langle n_1, n_2|\mc E_{\gamma,\Theta}(a_3\ten \mc E_{\gamma,\Theta}(\cdots \mc E_{\gamma,\Theta}(a_k\ten \1_{H})))) \\ 
&=& \sum_{n_1,n_2=1}^Np(n_1, n_2)\tr(|s_{n_1}\rangle\langle s_{n_1}|a_1)\tr(|s_{n_2}\rangle\langle s_{n_2}|a_2)\times \\ 
&\phantom{a}&\qquad  \tr(|n_1, n_2\rangle\langle n_1, n_2|\mc E_{\gamma,\Theta}(a_3\ten \mc E_{\gamma,\Theta}(\cdots \mc E_{\gamma,\Theta}(a_k\ten \1_{H})))) \\ 
& \vdots & \\
&=& \sum_{n_1,\ldots,n_k=1}^Np(n_1,\ldots, n_k)\tr(|s_{n_1}\rangle\langle s_{n_1}|a_1)\cdots \tr(|s_{n_k}\rangle\langle s_{n_k}|a_k),
\end{eqnarray*}
where the ``moreover" part of Lemma~\ref{lemma2} was used in the $3^{\text{rd}}$ equality, the fact $\tr(A\ten B)=\tr(A)\tr(B)$ was used in the $4^{\text{th}}$ equality and Lemma~\ref{lemma2} was used in the $5^{\text{th}}$ equality. 

Thus, for each $k\in\N$, the density matrix $\rho_k$ which is defined in Equation~\eqref{joint correlations} is given by 
\begin{equation*}\label{stochastic rhon=rhok}
\rho_k= \sum_{n_1,\ldots,n_k=1}^Np(n_1,\ldots, n_k) |s_{n_1}\cdots s_{n_k}\rangle \langle s_{n_1}\cdots s_{n_k}| = \rho_{S^k}.
\end{equation*} 
Therefore, 
$$
h(\Theta , \rho , \gamma) =\limsup_{k\ra\infty}\frac{1}{k}S(\rho_k) =\limsup_{k\ra\infty}\frac{1}{k}S(\rho_{S^k})=\limsup_{k\ra\infty} EL_k^*(\rho_{\mc S^k}),$$ where the first equality holds by the definition of the dynamical entropy in Equation~\eqref{dyn entropy 1} and the last equality follows from Equation~\eqref{bounds 2}. We have proved the following theorem. 

\begin{Thm}\label{main theorem 2}
	Given any stochastic ensemble $\{\mc S^k\}_{k=1}^\infty$, the optimal average codeword length per symbol (via lossless coding) converges to the dynamical entropy of the above-described quantum dynamical system $(\mc A, \Theta ,\rho)$ with respect to the operational partition of unity $\gamma$ defined by Equation~\eqref{E:GeneralGamma} in the following sense: 
$$ 
\limsup_{k\ra\infty} EL^*_k(\rho_{S^k})=\limsup_{k\ra\infty}\frac{1}{k}S(\rho_{\mc S^k})=h(\Theta ,\rho, \gamma).
$$
\end{Thm}

It should be noted that Theorem~\ref{main theorem} can be considered a corollary of Theorem~\ref{main theorem 2}. 
However, we have presented it separately since the construction is simpler in the case of Markov ensembles. 

\subsection{Examples}
Examples of quantum sources that produce not-necessarily 
statistically independent quantum symbols have been considered in the literature.
In \cite{KingL98} examples of quantum sources that produce 
statistically independent symbols (called ``Bernoulli sources")
as well as quantum sources producing not-necessarily 
statistically independent quantum symbols are considered.
In \cite{BF02} the authors consider quantum Morse codes 
as an example of quantum communication since quantum
data compression can be viewed as a special case of noiseless
quantum communication. In communications, either classical
or quantum, the assumption of statistical independence of
the symbols to be communicated, gives a serious restriction
to the content of information which is communicated. Thus the 
need of considering quantum sources emitting not-necessarily 
statistical independent quantum symbols, naturally arises.

We have already shown in Corollary~\ref{cor1} that we can recover the result of \cite{Schumacher95} and  \cite{BBHZ17} which states that the optimal codeword length per symbol for an i.i.d.\ prepared ensemble (i.e.\ Bernoulli sources)  is equal to the von Neumann entropy of the initial ensemble state. 
 First we illustrate that we can recover Theorem~\ref{expected per symbol} from Theorem~\ref{main theorem 2}. 

\begin{Ex}[Classical-Quantum Codes]
Let $S=\{n\}_{n=1}^d$ be a classical symbol set of cardinality $d$ for some $d\in\N$, $C:S\ra A^+$ be a uniquely decodable code into strings from the binary alphabet $A$, and $\X=(X_n)_{n=1}^\infty$ be a stochastic process governing the frequency of symbols from the symbol set $S$. Let $H_\mc S$ be a $d$-dimensional Hilbert space spanned by an orthonormal basis $\{|s_n\rangle\}_{n=1}^d$ and define the stochastic ensemble as usual by $\mc S^k=\{p(n_1,\ldots, n_k), |s_{n_1}\cdots s_{n_k}\rangle\}_{n_1,\ldots, n_k=1}^d$, where $p$ denotes the pmf of $\X$. Then since the $|s_n\rangle$'s are orthonormal it is easy to see that the ensemble states $\rho_{\mc S^k}$ are diagonal, for each $k\in\N$. Hence the optimal average codeword length per symbol for the stochastic ensemble, given by $h(\mc A, \Theta, \rho)$ in Theorem~\ref{main theorem 2}, is exactly equal to the entropy rate of the stochastic process (see Theorem~\ref{expected per symbol}). 
\end{Ex} 

Next we illustrate here the usefulness of Theorem~\ref{main theorem} on non-Bernoulli sources with two simple examples.
For each of the two examples, the Hilbert space 
$H_{\mc S}$ has 
dimension $d=2$ and an orthonormal basis 
$\{ | e_i \rangle \}_{i=1}^2$.

\begin{Ex}
For the second example consider the normalized non-orthogonal symbols $|s_1 \rangle = |e_1 \rangle$, 
$| s_2 \rangle = | e_2 \rangle$, 
$|s_3 \rangle= \frac{1}{\sqrt{2}}( | e_1 \rangle + |e_2 \rangle )$, 
(i.e. the Bell state $|+\rangle$), and 
$| s_4 \rangle = \frac{1}{\sqrt{2}} (|e_1 \rangle - |e_2 \rangle )$,
(i.e. the Bell state $|- \rangle$), which span $H_{\mc S}=\C^2$ (i.e. we consider $N=4$ in the setting
described in Subsection~\ref{Q Data Compression}). Consider the transition matrix 
$$
P= (p_{i,j})_{i,j =1}^4 = \left( 
\begin{array}{cccc}
\frac12 & \frac12 & 0 & 0  \\
\frac12 & \frac12 & 0 & 0 \\
0 & 0 & \frac12 & \frac{1}{2} \\
0 & 0 & \frac12 & \frac{1}{2} 
\end{array} \right)
$$
where $p_{i,j}$ represents the conditional probability that 
the quantum source emits $|s_i \rangle$ right after it emits 
$| s_j \rangle$. A (non-unique) fixed probability distribution of $P$
is equal to the column vector $( \frac14 \quad  \frac14 \quad \frac14 \quad \frac14 )^T$.
Consider the quantum dynamical system 
$( \mc A, \Theta , \rho )$ where $\mc A = B(\C^4)$,
$\Theta : \mc A \to \mc A$ is given by
$$
\Theta (|k \rangle \langle \ell | ) = \delta_{k, \ell} \sum_{i=1}^3
p_{\ell , i} |i \rangle \langle i | ,
$$
i.e.
$$
\Theta \left( (a_{i,j})_{i,j=1}^4 \right)
= \left( \begin{array}{cccc}
\frac12 (a_{1,1} + a_{2,2}) & \frac12 (a_{1,1}+ a_{2,2}) & 0 & 0 \\
\frac12 (a_{1,1} + a_{2,2}) & \frac12 (a_{1,1}+ a_{2,2}) & 0 & 0 \\
0 & 0 & \frac12 (a_{1,1} + a_{2,2}) & \frac12 (a_{3,3}+ a_{4,4})\\
0 & 0 & \frac12 (a_{1,1} + a_{2,2}) & \frac12 (a_{3,3}+ a_{4,4})
\end{array} \right),
$$
and
$$
\rho= \frac14 |1 \rangle \langle 1 | + \frac14 |2 \rangle \langle 2|
+ \frac14 | 3 \rangle \langle 3 | +\frac14 | 4 \rangle \langle 4 |
= \left( \begin{array}{cccc}
\frac14 & 0 & 0 & 0 \\ 
0 & \frac14 & 0 & 0 \\ 
0 & 0 & \frac14 & 0 \\
0 & 0 & 0 & \frac14
\end{array} \right) .
$$
Consider the operational partition of unity $\gamma = (\gamma_i )_{i=1}^2$ for $\C^4$ given as in Equation~\eqref{partition}. 
Theorem~\ref{main theorem} states that the optimal average 
codeword length per symbol via lossless coding is equal to the 
the dynamical entropy of $\Theta$ with respect to the 
partition $\gamma$ when measured using the state $\rho$,
i.e. $h(\Theta , \rho , \gamma )$. We can compute the joint 
correlations $(\rho_n)_{n=1}^\infty$ of this dynamical system
using Equation~(\ref{rhon=rhok}) to see that 
$$
\rho_1 =  \frac14 \sum_{n=1}^4  |s_n \rangle \langle s_n | 
= \left( \begin{array}{cc}
\frac12 & 0 \\ 0 & \frac12
\end{array} \right) ,
$$
and in general 
$$
\rho_k = \frac{1}{2^{k+1}} \sum_{n_1 , \ldots , n_k =1}^2 
| s_{n_1} \rangle \langle s_{n_1} | \otimes \cdots \otimes | s_{n_k} \rangle \langle s_{n_k} | + 
\frac{1}{2^{k+1}} \sum_{m_1 , \ldots , m_k =1}^2 
| s_{m_1} \rangle \langle s_{m_1} | \otimes \cdots \otimes | s_{m_k} \rangle \langle s_{m_k} | .
$$
It is easy to verify that
$$
\sum_{n_1 , \ldots , n_k =1}^2 
| s_{n_1} \rangle \langle s_{n_1} | \otimes \cdots \otimes | s_{n_k} \rangle \langle s_{n_k} | =
\sum_{m_1 , \ldots , m_k =3}^4 
| s_{m_1} \rangle \langle s_{m_1} | \otimes \cdots \otimes | s_{m_k} \rangle \langle s_{m_k} | 
= \1_{\C^{2^k}},
$$
by applying these sums to bases of $\C^{2^k}$ that are formed
by taking the tensor products of $|s_i \rangle$'s. Thus $\rho_k = \frac{1}{2^k} \1_{\C^{2^k}}$ and hence $S(\rho_k) =k$, for each $k\in\N$. Therefore by Theorem~\ref{main theorem} we obtain that the optimal average compression rate per symbol for the above quantum ensemble
is equal to 1 qubit.

\end{Ex}

\begin{Ex}
Consider the normalized 
non-orthogonal
symbol states $|s_1 \rangle =  |e_1 \rangle $, 
$|s_2 \rangle = - \frac12 |e_1 \rangle + \frac{\sqrt{3}}{2}|e_2 \rangle $
and $| s_3 \rangle = - \frac12 |e_1 \rangle - \frac{\sqrt{3}}{2}|e_2 \rangle $ which 
span $H_{\mc S}$ (i.e. we consider $N=3$ in the setting described in Subsection~\ref{Q Data Compression}). Consider the transition matrix 
$$
P= (p_{i,j})_{i,j =1}^3 = \left( 
\begin{array}{ccc}
0 & \frac12 & \frac12 \\
\frac12 & 0 & \frac12 \\
\frac12 & \frac12 & 0
\end{array} \right)
$$
where $p_{i,j}$ represents the conditional probability that 
the quantum source emits $|s_i \rangle$ right after it emits 
$| s_j \rangle$. The unique fixed probability distribution of $P$
is equal to the column vector $( \frac13 \quad  \frac13 \quad \frac13)^T$.
Consider the quantum dynamical system 
$( \mc A, \Theta , \rho )$ where $\mc A = B(\C^3)$,
$\Theta : \mc A \to \mc A$ is given by
$$
\Theta (|k \rangle \langle \ell | ) = \delta_{k, \ell} \sum_{i=1}^3
p_{\ell , i} |i \rangle \langle i | ,
$$
i.e.
$$
\Theta \left( (a_{i,j})_{i,j=1}^3 \right)
= \left( \begin{array}{ccc}
\frac12 (a_{2,2}+a_{3,3}) & 0 & 0 \\
0 & \frac12 (a_{1,1}+a_{3,3}) & 0 \\
0 & 0 & \frac12 (a_{1,1}+a_{3,3})
\end{array} \right),
$$
and
$$
\rho= \frac13 |1 \rangle \langle 1 | + \frac13 |2 \rangle \langle 2|
+ \frac13 | 3 \rangle \langle 3 | = \left( \begin{array}{ccc}
\frac13 & 0 & 0 \\ 0 & \frac13 & 0 \\ 0 & 0 & \frac13 
\end{array} \right) .
$$
Consider the operational partition of unity $\gamma = (\gamma_i )_{i=1}^2$ for $\C^3$ given by
$$
\gamma_1=\langle e_1 | s_1 \rangle |1 \rangle \langle 1| +
\langle e_1 |s_2 \rangle | 2 \rangle \langle 2| + \langle e_1 | s_3 \rangle |3 \rangle \langle 3| =
\left( \begin{array}{ccc}
1 & 0 & 0 \\ 0 & -\frac12 & 0 \\ 0 & 0 & -\frac12 
\end{array}  \right)
$$
and 
$$
\gamma_2=\langle e_2 | s_1 \rangle |1 \rangle \langle 1| +
\langle e_2 |s_2 \rangle | 2 \rangle \langle 2| + \langle e_2 | s_3 \rangle | 3 \rangle \langle 3 | =
\left( \begin{array}{ccc}
0 & 0 & 0 \\ 0 & \frac{\sqrt{3}}{2} & 0 \\ 0 & 0 & \frac{\sqrt{3}}{2}
\end{array}  \right) .
$$
Theorem~\ref{main theorem} states that the optimal average 
codeword length per symbol via lossless coding is equal to the 
the dynamical entropy of $\Theta$ with respect to the 
partition $\gamma$ when measured using the state $\rho$,
i.e. $h(\Theta , \rho , \gamma )$. We can compute the joint 
correlations $(\rho_n)_{n=1}^\infty$ of this dynamical system
using Equation~(\ref{rhon=rhok}) to see that
$$
\rho_1 =  \frac13 |s_1 \rangle \langle s_1 | + \frac13 | s_2 \rangle \langle s_2 | + \frac13 | s_3 \rangle \langle s_3 | 
= \left( \begin{array}{cc}
\frac12 & 0 \\ 0 & \frac12
\end{array} \right) ,
$$
and in general 
$$
\rho_k = \frac{1}{3 \cdot 2^{k-1}}
\sum_{n_1=1}^3 \sum_{n_2,\ldots , n_k=1}^2 
|s_{n_1}\rangle \langle s_{n_1}| \otimes | s_{n_2'} \rangle \langle s_{n_2'} | \otimes \cdots \otimes 
| s_{n_k'} \rangle \langle s_{n_k'} | ,
$$
where
$$
n_k'=\sum_{l=1}^k n_l\mod 3
$$
and we adopt the convention that the mod~3 function takes values in the set $\{ 1, 2, 3\}$. Using Matlab we can obtain 
the following approximate values of the 
von~Neumann entropies of the above matrices:
\begin{center}
\begin{tabular}{|c|c|c|c|c|c|c|}
\hline
$k$ & 1 & 2 & 3 & 4 & 5 & 6 \\ 
\hline 
$\frac{1}{k} S(\rho_k)$ & 1 & 0.9528  & 0.9306 & 0.9169 & 0.9076 & 0.9008 \\
\hline \hline
$k$ & 7 & 8 & 9 &10 &11 &12 \\
\hline
$\frac{1}{k} S(\rho_k)$ & 0.8957 & 0.8918 & 0.8886 & 0.8861 & 0.8839 & 0.8822 \\
\hline
\end{tabular}
\end{center}

\noindent
The above decreasing numbers indicate that  the optimal average compression rate per symbol for the above quantum stochastic ensemble is strictly less than 1 qubit.

\end{Ex}

\section{Concluding Remarks}

In this paper, we developed further the theory of quantum data compression for indeterminate length quantum codes, building on the previous work of Schumacher and Westmoreland \cite{SW01} and Bellomo, et.\ al \cite{BBHZ17}. 
We presented the quantum Kraft inequality with an additional converse statement which was not present in previous works; this additional converse statement makes the statement of the quantum Kraft inequality more reminiscent of its classical counterpart. 
We also introduce the notion of stochastic ensembles and, in particular, stationary Markov ensembles which, to the best of our knowledge, have not been considered elsewhere. 
The main contributions of this work are Theorems~\ref{main theorem} and \ref{main theorem 2} which give a dynamical entropy interpretation of the optimal compression rate for stationary Markov and identically distributed stochastic ensembles, respectively, 
extending the results of Schumacher \cite{Schumacher95} and Bellomo et.\ al \cite{BBHZ17} where the quantum symbol states to be encoded were prepared in an i.i.d.\ way.  
In doing so, we give a quantum Markov chain representation of a particular open quantum random walk. 
An interesting direction for future study is the development of quantum data compression on the symmetric Fock space which is commonly used to model photons. We hope to develop this theory further in future work.

\pagebreak

\section{Appendix: Proofs of auxiliary results}

\begin{proof}[Proof of Theorem~\ref{quantum Kraft}:]
For the forward direction we adapt the proof of \cite[Subsection II.C.]{SW01} to our formalism. 
Let $U$ be a uniquely decodable quantum code with length eigenstates of the form $$U=\sum_{i=1}^d |\psi_i\rangle\langle e_i|$$ 
and let $\{\ell_i\}_{i=1}^d$ be the length eigenvalues of $U$. 
For each $n, N\in \N$, let 
$$C_n^N=\{|\psi\rangle \in H_{\mc A}^{\ten N}:|\psi\rangle =|\psi_{i_1}\rangle |\psi_{i_2}\rangle\cdots |\psi_{i_n}\rangle\text{ for some }i_1,\ldots, i_N\in\{1,\ldots, d\}\}$$ be the collection of length $N$ strings consisting of $n$-many codewords and let  
$$d_{\ell}=|\{i\in\{1,\ldots, d\}:\psi_i\in H_{\mc A}^{\ten \ell}\}|=|\{i\in\{1,\ldots,d\}: \ell_i=\ell\}|$$ be the number of length $\ell$ eigenstates of $U$, for each $\ell\in\N$. 
Then, by the unique decodability of $U$, each element of $C_n^N$ has a unique representation as a string of $n$ codewords and the elements of $C_n^N$ are pairwise orthogonal, and hence we have 
$$|C_n^N|=\sum_{\ell_{i_1}+\cdots +\ell_{i_n}=N} d_{\ell_{i_1}}d_{\ell_{i_2}}\cdots d_{\ell_{i_n}}\le 2^N.$$ 
Thus 
$$2^{-N}\sum_{\ell_{i_1}+\cdots +\ell_{i_n}=N} d_{\ell_{i_1}}d_{\ell_{i_2}}\cdots d_{\ell_{i_n}} 
= \sum_{\ell_{i_1}+\cdots +\ell_{i_n}=N} (2^{-\ell_{i_1}}d_{\ell_{i_1}})(2^{-\ell_{i_2}}d_{\ell_{i_2}})\cdots (2^{-\ell_{i_n}}d_{\ell_{i_n}})\le 1.$$ 
Set $\ds{\lmax=\max_{1\le i\le d}\{\ell_i\}}$ so that $N\le n\lmax$. 
Summing the above inequality over $N$ we obtain 
$$\sum_{\ell_{i_1},\ell_{i_2},\cdots,\ell_{i_n}=1}^{\ell_{\text{max}}}(2^{-\ell_{i_1}}d_{\ell_{i_1}})(2^{-\ell_{i_2}}d_{\ell_{i_2}})\cdots (2^{-\ell_{i_n}}d_{\ell_{i_n}})=\left( \sum_{\ell=1}^{\ell_{\text{max}}}(2^{-\ell}d_{\ell})\right)^n \le n\ell_{\text{max}}.$$ 
Notice that the left-hand side of this inequality is exponential whereas the right-hand side is linear. This implies that the left-hand side is bounded above by 1. Hence we must have that 
\begin{equation}\label{Kraft inequality}
\tr(U^\dagger 2^{-\Lambda}U)=\sum_{\ell=1}^{\ell_{\text{max}}}2^{-\ell} \tr (U^\dagger \Pi_\ell U)=\sum_{\ell=1}^{\ell_{\text{max}}}2^{-\ell}d_{\ell}\le 1.
\end{equation}
Notice that the inequality in Equation~\eqref{Kraft inequality} is simply a restatement of the classical Kraft-McMillan inequality. 

Conversely, suppose that $U$ is a linear isometry with length eigenstates satisfying the quantum Kraft-McMillan Inequality, and define $\{\ell_i\}_{i=1}^d$, $\lmax$ and $\{d_{\ell} \}_{\ell=1}^{\lmax}$ as above. 
Then 
\begin{equation*}
\sum_{\ell=1}^{\ell_{\text{max}}}2^{-\ell}d_{\ell}=\sum_{\ell=1}^{\ell_{\text{max}}}2^{-\ell} \tr (U^\dagger \Pi_\ell U)= \tr(U^\dagger 2^{-\Lambda}U)\le 1 
\end{equation*}
and hence the classical Kraft-McMillan inequality is also valid. 
Thus, by the converse of the classical Kraft-McMillan Theorem, one can find a classical uniquely decodable code $C$ which has exactly $d_\ell$-many codewords of length $\ell$, for each $\ell\in\N$. The c-q scheme $\widetilde U$ constructed from this classical code $C$ has the desired properties. 
\end{proof}

\begin{proof}[Proof of Lemma~\ref{lemma1}:]
Since $\mc E_{\gamma,\Theta}= \Theta \circ \mc E_{\gamma}$ 
we have that $\mc E_{\gamma,\Theta}^\dagger = \mc E_{\gamma}^\dagger \circ \Theta^\dagger$.

Next we consider the lifting $\mc E_{\gamma}^\dagger : S_1(\C^N)\ra S_1(H_C)\ten S_1(\C^N)$ which we claim is given by the formula
\begin{equation}\label{lifting}
\mc E_{\gamma}^\dagger (\sigma)= [\gamma_i \sigma \gamma_j^*]_{i,j=1}^d=\sum_{i,j=1}^d |e_i\rangle\langle e_j|\ten \gamma_i \sigma \gamma_j^*,
\end{equation}
where we have identified $S_1(H_{\mc S})$ with $M_d$ given the matrix representation with respect to the fixed orthonormal basis $\{|e_i\rangle\}_{i=1}^d$ used in Equations~\eqref{trans exp 1} and \eqref{partition}.  
Indeed, for $[a_{i,j}]_{i,j=1}^d\in B(H_{\mc S})\ten B(\C^N)$ and $\sigma\in S_1(\C^N)$, we have 
\begin{equation*}
\begin{split}
\tr(\sigma \mc E_\gamma ([a_{i,j}]_{i,j=1}^d)) &= \tr(\sigma \sum_{i,j=1}^d \gamma_i^* a_{i,j}\gamma_j) 
=\sum_{i,j=1}^d\tr(\sigma \gamma_i^* a_{i,j}\gamma_j) \\
&= \sum_{i,j=1}^d \tr(\gamma_j\sigma \gamma_i^* a_{i,j}) 
\hspace{.04in}= \sum_{i,j=1}^d \tr(\gamma_i\sigma \gamma_j^* a_{j,i}) \\
&= \tr( \sum_{i,j=1}^d |e_i\rangle\langle e_j|\ten \gamma_i \sigma \gamma_j^* [a_{i, j}]_{i,j=1}^d) 
\end{split}
\end{equation*} 
which proves the validity of Equation~\eqref{lifting}. 
Then, for each $|m\rangle\langle m|\in S_1(\C^N)$, we have 
\begin{equation}\label{restofway}
\begin{split}
\mc E_{\gamma}^\dagger (|m\rangle\langle m|) 
&= \sum_{i,j=1}^d |e_i\rangle\langle e_j|\ten \gamma_j |m\rangle\langle m| \gamma_i^* \\
&= \sum_{i,j=1}^d |e_i\rangle\langle e_j|\ten \langle e_i, s_m\rangle |m\rangle\langle m| \langle s_m, e_j\rangle \quad\text{by Equation~\eqref{partition}} \\
&= \left|\sum_{i=1}^d \langle e_i, s_m\rangle e_i\right\rangle \left\langle\sum_{j=1}^d \langle e_j, s_m\rangle e_j \right|\ten |m\rangle\langle m| \\
&= |s_m\rangle\langle s_m|\ten |m\rangle\langle m| = |s_m'\rangle\langle s_m'|.
\end{split}
\end{equation}

Combining Equations~\eqref{E:MarkovTheta} and \eqref{restofway}, for each $|n\rangle\langle n|\in S_1(\C^N)$, we have 
\begin{equation}\label{lastpart}
\begin{split}
\mc E_{\gamma,\Theta}^\dagger (|n\rangle\langle n|) 
&= \mc E_{\gamma}^\dagger (\sum_{m=1}^N p_{m,n} |m\rangle\langle m|)\quad\text{by Equation~\eqref{E:MarkovThetaDagger}} \\ 
&= \sum_{m=1}^N p_{m,n} |s_m'\rangle\langle s_m'|\quad\text{by Equation~\eqref{restofway}.}
\end{split}
\end{equation}

For the moreover statement, we have 
\begin{equation*}\label{inv state}
\begin{split}
\mc E_{\gamma,\Theta}^\dagger (\rho) 
&= \sum_{n=1}^N p_n \mc E_{\gamma,\Theta}^\dagger (|n\rangle\langle n|) \quad\text{by Equation~\eqref{E:MarkovRho}} \\ 
&= \sum_{n,m=1}^N p_n p_{m,n} |s_m'\rangle\langle s_m'| \quad\text{by Equation~\eqref{lastpart}} \\ 
&= \sum_{m=1}^N p_m |s_m'\rangle\langle s_m'|= \rho' \quad\text{since $\X$ is stationary.}
\end{split}
\end{equation*}
\end{proof}

\begin{proof}[Proof of Lemma~\ref{lemma2}:]
It is easy to see that for each $\bar m , \bar n \in \{ 1, \ldots , N \}^+$ (i.e. \ $| \bar m \rangle, | \bar n \rangle $ belong in the standard 
orthonormal basis of $H$), we have 
	\begin{equation}\label{stochastic halfway}
	\Theta^\dagger (|\bar m \rangle\langle \bar n|) = 
\delta_{\bar m , \bar n}	\sum_{k=1}^N p(k|\bar n) |\bar n\circ k\rangle\langle \bar n\circ k| .
	\end{equation}

	Next we consider the lifting $\mc E_{\gamma}^\dagger : S_1(H)\ra S_1(H_{\mc S})\ten S_1(H)$ which (by Equation~\eqref{lifting}) is given by the formula
	$$
	\mc E_{\gamma}^\dagger (\sigma)= [\gamma_i \sigma \gamma_j^*]_{i,j=1}^d=\sum_{i,j=1}^d |e_i\rangle\langle e_j|\ten \gamma_i \sigma \gamma_j^*,
	$$ 
	where we have identified $S_1(H_{\mc S})$ with $M_d$ given the matrix representation with respect to the fixed orthonormal basis $\{|e_i\rangle\}_{i=1}^d$ of the Hilbert space 
	$H_{\mc S}$.
	Then, for each  $\bar n \in \{ 1, \ldots , N \}^+$, we have 
	\begin{equation}\label{stochastic restofway}
	\begin{split}
	\mc E_{\gamma}^\dagger ( |\bar n\rangle\langle \bar n|) 
	&= \sum_{i,j=1}^d |e_i\rangle\langle e_j|\ten \gamma_i  |\bar n\rangle\langle \bar n| \gamma_j^* \\
	&= \sum_{i,j=1}^d |e_i\rangle\langle e_j|\ten \langle e_i | s_{\final (\bar n)}\rangle |\bar n\rangle\langle \bar n| \langle s_{\final (\bar n)} | e_j\rangle \\
	&= \left|\sum_{i=1}^d \langle e_i| s_{\final( \bar n )}\rangle e_i\right\rangle \left\langle\sum_{j=1}^d \langle e_j | s_{\final( \bar n) }\rangle e_j \right|\ten |\bar n\rangle\langle \bar n| \\
	&= |s_{\final (\bar n )}\rangle\langle s_{\final (\bar n )}|\ten  |\bar n\rangle\langle \bar n| = |s_{\bar n}'\rangle\langle s_{\bar n}'|,
	\end{split}
	\end{equation}
	where we used Equation~\eqref{E:GeneralGamma} in the second equality.

	Combining Equations~\eqref{stochastic halfway} and \eqref{stochastic restofway}, for each   $\bar n \in \{ 1, \ldots , N \}^+$, we have 
	\begin{equation}\label{stochastic lastpart}
	\begin{split}
	\mc E_{\gamma,\Theta}^\dagger(|\bar n\rangle\langle \bar n|) 
	&= \mc E_{\gamma}^\dagger\left(\sum_{k=1}^N p(k|\bar n) |\bar n\circ k\rangle\langle \bar n\circ k|\right) \\
	&= \sum_{k=1}^N p(k|\bar n) |s_{\bar n\circ k}'\rangle\langle s_{\bar n\circ k}'|. 
	\end{split}
	\end{equation}
	
	For the moreover statement, we have 
	\begin{equation*}\label{stochastic inv state}
	\begin{split}
	\mc E_{\gamma,\Theta}^\dagger (\rho)
	&= \sum_{k=1}^N p(k|\emptyset) |s_{k}'\rangle\langle s_{k}'| \quad\text{by Equation~\eqref{stochastic lastpart}} \\ 
	&= \sum_{k=1}^N p(k) |s_{k}'\rangle\langle s_{k}'|, 
	\end{split}
	\end{equation*}
	where we again used the convention that $p(k|\emptyset)=p(k)$, for all $k\in\{1,\ldots, N\}$, in the last equality. 
\end{proof}

\end{document}